\documentclass[proceedings]{stacs}
\stacsheading{2009}{661--672}{Freiburg}
\firstpageno{661}

\usepackage[latin1]{inputenc}
\usepackage[usenames]{color}
\usepackage{amsmath}
\usepackage{amsfonts}
\usepackage{amssymb}
\usepackage{stmaryrd}
\usepackage{framed}
\usepackage{pstricks, pst-node, pst-tree}

\theoremstyle{plain}
\theoremstyle{definition}

\newenvironment{construction}{
        \noindent {\bf Construction: }}


\begin{document}
\sloppy

\newcommand\buchi{B\"uchi}
\newcommand\maxi{\mathit{max}}
\newcommand\odd{\mathit{odd}}
\newcommand\rank{\mathit{rank}}
\newcommand\te{\mathit{th}}
\newcommand\tight{\mathit{tight}}

\title{B\"uchi Complementation Made Tight}

\author{Sven Schewe}{Sven Schewe}
\address{University of Liverpool}
\email{Sven.Schewe@liverpool.ac.uk}
\urladdr{http://www.csc.liv.ac.uk/research/logics/}

\thanks{This work was partly supported by the EPSRC through the grand EP/F033567/1 \emph{Verifying Interoperability Requirements in Pervasive Systems}}

\def\thisbottomragged{\def\@textbottom{\vskip\z@ plus.0001fil
\global\let\@textbottom\relax}}

\frenchspacing
\widowpenalty=10000
\clubpenalty=10000

\bibliographystyle{alpha}
\maketitle

\begin{abstract}
The precise complexity of complementing \buchi\ automata is an intriguing and long standing problem.
While optimal complementation techniques for finite automata are simple -- it suffices to determinize them using a simple subset construction and to dualize the acceptance condition of the resulting automaton -- \buchi\ complementation is more involved.
Indeed, the construction of an EXPTIME complementation procedure took a quarter of a century from the introduction of \buchi\ automata in the early $60$s, and stepwise narrowing the gap between the upper and lower bound to a simple exponent (of $(6e)^n$ for \buchi\ automata with $n$ states) took four decades.
While the distance between the known upper ($O\big((0.96\,n)^n\big)$) and lower ($\Omega\big((0.76\,n)^n\big)$) bound on the required number of states has meanwhile been significantly reduced, an exponential factor remains between them.
Also, the upper bound on the size of the complement automaton is not linear in the bound of its state space.
These gaps are unsatisfactory from a theoretical point of view, but also because \buchi\ complementation is a useful tool in formal verification, in particular for the language containment problem.
This paper proposes a \buchi\ complementation algorithm whose complexity meets, modulo a quadratic ($O(n^2)$) factor, the known lower bound for \buchi\ complementation.
It thus improves over previous constructions by an exponential factor and concludes the quest for optimal \buchi\ complementation algorithms.
\end{abstract}

\section{Introduction}

The precise complexity of \buchi\ complementation is an intriguing problem for two reasons:
First, the quest for optimal algorithms is a much researched problem (c.f.,~\cite{Buchi/62/Automata,Sakoda+Sipser/78/finTight,Pecuchet/86/Complement,Sistla+Vardi+Wolper/87/Complement,Safra/88/Safra,Michel/88/lowerComplementation,Thomas/99/complementation,Loding/99/omega,Kupferman+Vardi/01/Weak,GKSV/03/Complementing,FKV/06/tighter,Piterman/07/Parity,Vardi/07/Saga,Yan/08/lowerComplexity}) that has defied numerous approaches to solving it.
And second, \buchi\ complementation is a valuable tool in formal verification (c.f.,~\cite{kurshan/94/verification}), in particular when studying language inclusion problems of $\omega$-regular languages.
In addition to this, complementation is useful to check the correctness of other translation techniques~\cite{Vardi/07/Saga,Tsay/08/goal}.
The GOAL tool~\cite{Tsay/08/goal}, for example, provides such a test suite and incorporates four of the more recent algorithms~\cite{Safra/88/Safra,Thomas/99/complementation,Kupferman+Vardi/01/Weak,Piterman/07/Parity} for \buchi\ complementation.

While devising optimal complementation algorithms for nondeterministic finite automata is simple---%
nondeterministic \emph{finite} automata can be determinized using a simple subset construction, and deterministic finite automata can be complemented by complementing the set of final states~\cite{Rabin+Scott/59/finite,Sakoda+Sipser/78/finTight}---%
devising optimal complementation algorithms for nondeterministic \buchi\ automata is hard, because simple subset constructions are not sufficient to determinize or complement them~\cite{Michel/88/lowerComplementation,Loding/99/omega}.

Given the hardness and importance of the problem, \buchi\ complementation enjoyed much attention~\cite{Buchi/62/Automata,Pecuchet/86/Complement,Sistla+Vardi+Wolper/87/Complement,Michel/88/lowerComplementation,Safra/88/Safra,Loding/99/omega,Thomas/99/complementation,Kupferman+Vardi/01/Weak,GKSV/03/Complementing,FKV/06/tighter,Vardi/07/Saga,Tsay/08/goal,Yan/08/lowerComplexity}, resulting in a continuous improvement of the upper and lower bounds.

The first complementation algorithm dates back to the introduction of \buchi\ automata in 1962.
In his seminal paper ``On a decision method in restricted second order arithmetic'' \cite{Buchi/62/Automata}, \buchi\ develops a complementation procedure that comprises a doubly exponential blow-up.
While \buchi's result shows that nondeterministic \buchi\ automata (and thus \mbox{$\omega$-regular} expressions) are closed under complementation, complementing an automaton with $n$ states may, when using \buchi's complementation procedure, result in an automaton with $2^{2^{O(n)}}$ states, while an $\Omega(2^n)$ lower bound~\cite{Sakoda+Sipser/78/finTight} is inherited from finite automata.

In the late 80s, these bounds have been improved in a first sequence of results, starting with establishing an EXPTIME upper bound~\cite{Pecuchet/86/Complement,Sistla+Vardi+Wolper/87/Complement}, which matches the EXPTIME lower bound~\cite{Sakoda+Sipser/78/finTight} inherited from finite automata.
However, the early EXPTIME complementation techniques produce automata with up to $2^{O(n^2)}$ states~\cite{Pecuchet/86/Complement,Sistla+Vardi+Wolper/87/Complement};
hence, these upper bounds are still exponential in the lower bounds.

This situation changed in 1988, when Safra introduced his famous determinization procedure for nondeterministic \buchi\ automata~\cite{Safra/88/Safra}, resulting in an $n^{O(n)}$ bound for \buchi\ complementation, while Michel~\cite{Michel/88/lowerComplementation} established a seemingly matching $\Omega(n!)$ lower bound in the same year.
Together, these results imply that \buchi\ complementation is in $n^{\theta(n)}$, leaving again the impression of a tight bound.

As pointed our by Vardi~\cite{Vardi/07/Saga}, this impression is misleading, because the $O()$ notation hides an $n^{\theta(n)}$ gap between both bounds.
This gap has been narrowed down in 2001 to $2^{\theta(n)}$ by the introduction of an alternative complementation technique that builds on level rankings and a cut-point construction~\cite{Kupferman+Vardi/01/Weak}. (Level rankings are functions from the states $Q$ of a nondeterministc \buchi\ automaton to $\{0,1,\ldots,2|Q|+1\}$.)
The complexity of the plain method is approximately $(6n)^n$~\cite{Kupferman+Vardi/01/Weak}, leaving a $(6e)^n$ gap to Michel's lower bound~\cite{Michel/88/lowerComplementation}.

Recently, tight level rankings~\cite{FKV/06/tighter,Yan/08/lowerComplexity}---a special class of level rankings that is onto a predefined subset---have been exploited
by Friedgut, Kupferman, and Vardi~\cite{FKV/06/tighter} to improved the upper complexity bound to $O\big((0.96 n)^n\big)$, and
by Yan~\cite{Yan/08/lowerComplexity} to improve the lower complexity bound to $\Omega\big((0.76 n)^n\big)$.

In the remainder of this paper, we first recapitulate the basic complementation technique of Kupferman and Vardi~\cite{Kupferman+Vardi/01/Weak},
and discuss the core ideas of the improved complexity analysis of Friedgut, Kupferman, and Vardi~\cite{FKV/06/tighter} and Yan~\cite{Yan/08/lowerComplexity}.
We then show how to improve the complementation technique of Friedgut, Kupferman, and Vardi~\cite{FKV/06/tighter} such that the resulting complementation algorithm meets the known lower bound~\cite{Yan/08/lowerComplexity} modulo a small polynomial factor (quadratic in the size of the automaton that is to be complemented),
and show that, different to older constructions~\cite{Kupferman+Vardi/01/Weak,GKSV/03/Complementing}, we can achieve an equivalent bound on the number of edges.

\section{Preliminaries}
\subsection{\buchi\ Automata}

Nondeterministic \buchi\ automata~\cite{Buchi/62/Automata} are used to represent $\omega$-regular languages $L\subseteq \Sigma^\omega = \omega \rightarrow \Sigma$ over a finite alphabet $\Sigma$.
A nondeterministic \buchi\ automaton $\mathcal A= (\Sigma,Q,I,\delta,F)$ is a five tuple, consisting of a finite alphabet $\Sigma$,
a finite set $Q$ of states with a non-empty subset $I\subseteq Q$ of initial states,
a transition function $\delta: Q\times \Sigma \rightarrow 2^Q$ that maps states and input letters to sets of successor states, and
a set $F\subseteq Q$ of final states.

Nondeterministic \buchi\ automata are interpreted over infinite sequences $\alpha: \omega \rightarrow \Sigma$ of input letters.
An infinite sequence $\rho: \omega \rightarrow Q$ of states of $\mathcal A$ is called a \emph{run} of $\mathcal A$ on an input word $\alpha$ if the first letter $\rho(0) \in I$ of $\rho$ is an initial state, and if, for all $i\in \omega$, $\rho(i+1)\in \delta\big(\rho(i),\alpha(i)\big)$ is a successor state of $\rho(i)$ for the input letter $\alpha(i)$.

A run $\rho: \omega \rightarrow Q$ is called \emph{accepting} if some finite state appears infinitely often in~$\rho$ ($\mathit{inf}(\rho)\cap F\neq \emptyset$ for $\mathit{inf}(\rho)=\{q\in Q \mid \forall i \in \omega \; \exists j>i\mbox{ such that } \rho(j)=q\}$).
A word $\alpha: \omega \rightarrow \Sigma$ is \emph{accepted} by $\mathcal A$ if $\mathcal A$ has an accepting run on $\alpha$, and
the set $\mathcal{L(A)}=\{\alpha \in \Sigma^\omega \mid \alpha \mbox{ is accepted by }\mathcal A\}$ of words accepted by $\mathcal A$ is called its \emph{language}.

For technical convenience we also allow for finite runs $q_0q_1q_2\ldots q_n$ with $\delta\big(q_n,\alpha(n)\big)=\emptyset$.
Naturally, no finite run satisfies the \buchi\ condition;
all finite runs are therefore rejecting, and have no influence on the language of an automaton.

The two natural complexity measures for a \buchi\ automaton are the size $|Q|$ of its state space, and its size $\sum\limits_{q\in Q,\ \sigma \in \Sigma}1+|\delta(q,\sigma)|$, measured in the size of its transition function.

\subsection{Run DAG and Acceptance}
\label{ssec:dag}

In~\cite{Kupferman+Vardi/01/Weak}, Kupferman and Vardi introduce a \buchi\ complementation algorithm that uses level rankings as witnesses for the absence of an accepting run.

The set of all runs of a nondeterministic \buchi\ automaton $\mathcal A=(\Sigma,Q,I,\delta,F)$ on a word $\alpha: \omega \rightarrow \Sigma$ can be represented by a directed acyclic graph (DAG) $\mathcal G_\alpha=(V,E)$ with
\begin{itemize}
\item vertices $V \subseteq Q\times \omega$ such that $(q,p)\in V$ is in the set $V$ of vertices if and only if there is a run $\rho$ of $\mathcal A$ on $\alpha$ with $\rho(p)=q$, and

\item edges $E \subseteq (Q\times \omega) \times (Q\times \omega)$ such that $\big((q,p),(q',p')\big)\in E$ if and only if $p'=p+1$ and $q' \in \delta\big(q,\alpha(p)\big)$ is a successor of $q$ for the input letter $\alpha(p)$.
\end{itemize}

We call $\mathcal G_\alpha=(V,E)$ the \emph{run DAG} of $\mathcal A$ for $\alpha$, and the vertices $V \cap (Q \times \{i\})$ of $\mathcal G_\alpha=(V,E)$ that refer to the $i^\te$ position of runs the \emph{$i^\te$ level} of $\mathcal G_\alpha=(V,E)$.

The run DAG $\mathcal G_\alpha$ is called \emph{rejecting} if no path in $\mathcal G_\alpha$ satisfies the \buchi\ condition.
That is, $\mathcal G_\alpha$ is rejecting if and only if $\mathcal A$ rejects $\alpha$.
$\mathcal A$ can therefore be complemented to a nondeterministc \buchi\ automaton $\mathcal B$ that checks if $\mathcal G_\alpha$ is rejecting.

The property that $\mathcal G_\alpha$ is rejecting can be expressed in terms of ranks.
We call a vertex $(q,p)\in V$ of a DAG $\mathcal G=(V,E)$ \emph{finite}, if the set of vertices reachable from $(q,p)$ in $\mathcal G$ is finite, and \emph{endangered}, if no vertex reachable from $(q,p)$ is accepting (that is, in $F\times \omega$).

Based on these definitions, \emph{ranks} can be assigned to the vertices of a rejecting run DAG.
We set ${\mathcal G_\alpha}^0= \mathcal G_\alpha$, and repeat the following procedure until a fixed point is reached, starting with $i=0$:

\begin{itemize}
\item Assign all finite vertices of ${\mathcal G_\alpha}^i$ the rank $i$, and set ${\mathcal G_\alpha}^{i+1}$ to ${\mathcal G_\alpha}^i$ minus the states with rank $i$ (that is, minus the states finite in ${\mathcal G_\alpha}^i$).

\item Assign all endangered vertices of ${\mathcal G_\alpha}^{i+1}$ the rank $i+1$, and set ${\mathcal G_\alpha}^{i+2}$ to ${\mathcal G_\alpha}^{i+1}$ minus the states with rank $i+1$ (that is, minus the states endangered in ${\mathcal G_\alpha}^{i+1}$).

\item Increase $i$ by $2$.
\end{itemize}

A fixed point is reached in $n+1$ steps, and the ranks can be used to characterize the complement language of a nondeterministic \buchi\ automaton:

\begin{proposition}
\cite{Kupferman+Vardi/01/Weak}
\label{prop:dag}
A nondeterministic \buchi\ automaton $\mathcal A$ with $n$ states rejects a word $\alpha: \omega \rightarrow \Sigma$ if and only if ${\mathcal G_\alpha}^{2n+1}$ is empty.
\qed
\end{proposition}

To see that a fixed point is reached after $n+1$ iterations,
note that deleting all finite or endangered vertices leaves a DAG without finite or endangered vertices, respectively.
Hence a fixed point is reached as soon as we do not assign a rank $i$ to any vertex.
By construction, no DAG ${\mathcal G_\alpha}^{2i+1}$ contains finite vertices.
If it contains an endangered vertex~$v$, then all vertices reachable from $v$ are endangered, too.
This implies that some vertex of almost all levels is assigned the rank $2i+1$.
Hence, if no fixed point is reached earlier, some rank is assigned to all vertices of almost all levels after $n$ iterations (there cannot be more than $n$ vertices in a level), and all vertices in ${\mathcal G_\alpha}^{2n}$ must be finite.

If the reached fixed point ${\mathcal G_\alpha}^\infty$ is non-empty, then it contains only infinite and non-endangered vertices, and
constructing an accepting run from ${\mathcal G_\alpha}^\infty$ is simple.
Vice versa, an accepting run $\rho$ on an $\omega$-word $\alpha$ can be viewed as an infinite sub-graph of $\mathcal G_\alpha$ that does not contain finite or endangered nodes. By a simple inductive argument, the sub-graph identified by $\rho$ is therefore a sub-graph of ${\mathcal G_\alpha}^i$ for all $i\in \omega$, and hence of ${\mathcal G_\alpha}^\infty$.

\subsection{\buchi\ Complementation}

The connection between \buchi\ complementation, run DAGs and ranks leads to an elegant complementation technique.
We call the maximal rank of a vertex in a level the rank of this level, the rank of almost all vertices $(\rho(i),i)$, $i \in \omega$ of a run $\rho$ (or: path in $\mathcal G$) the rank of $\rho$, and the rank of almost all levels of a DAG $\mathcal G$ the rank of $\mathcal G$. (Note that level ranks can only go down, and that vertex ranks can only go down along a path.)

For a given nondeterministic \buchi\ automaton $\mathcal A=(\Sigma,Q,I,\delta,F)$ with $n$ states, we call a function $f:Q \rightarrow \{0,1,\ldots,2n\}$ that maps all accepting states to odd numbers ($f(F) \cap 2\omega = \emptyset$) a \emph{level ranking}.

\begin{proposition}
\cite{Kupferman+Vardi/01/Weak}
\label{prop:basic}
For a given nondeterministic \buchi\ automaton $\mathcal A=(\Sigma,Q,I,\delta,F)$, the nondeterministic \buchi\ automaton $\mathcal B=(\Sigma,Q',I',\delta',F')$ with
\begin{itemize}
\item $Q' = 2^Q \times 2^Q \times \mathcal R$,

\item $I' = \{I\} \times \{\emptyset\} \times \mathcal R$,

\item $\delta'\big((S,O,f),\sigma)=\big\{\big(\delta(S,\sigma),\delta(O,\sigma)\smallsetminus \odd(f'),f'\big) \mid f' \leq^S_\sigma f,\ O\neq \emptyset \big\}$ \newline
\hspace*{24.7mm} $\cup \, \big\{\big(\delta(S,\sigma),\delta(S,\sigma)\smallsetminus \odd(f'),f'\big) \mid f' \leq^S_\sigma f,\ O= \emptyset \big\}$, and

\item $F' = 2^Q  \times \{\emptyset\} \times \mathcal R$,
\end{itemize}
where
\begin{itemize}
\item $\mathcal R$ is the set of all level rankings of $\mathcal A$,

\item $\odd(f) = \{q\in Q \mid f(q) \mbox{ is odd}\}$, and

\item $f'\leq^S_\sigma f \, :\Leftrightarrow \, \forall q \in S,\ q' \in \delta(q,\sigma).\ f'(q')\leq f(q)$,
\end{itemize}
accepts the complement $\mathcal{L(B)} = \overline{\mathcal{L(A)}} = \Sigma^\omega \smallsetminus \mathcal{L(A)}$ of the language of $\mathcal A$.
\qed
\end{proposition}

The first element $S$ of a triple $(S,O,f) \in Q'$ reflects the set of states of $\mathcal A$ reachable upon the input seen so far (the states in the respective level of the run DAG), and the third element $f$ is a mapping that intuitively maps reachable states to their rank.

The condition $f'\leq^S_\sigma f$ ensures that the rank of the vertices (or rather: the value assigned to them) is decreasing along every path of the run DAG.
The second element is used for a standard cut-point construction, comparable to the cut-point constructions in the determinization of Co-\buchi\ automata or the nondeterminization of alternating \buchi\ automata.
It contains the positions whose rank (or rather: the value assigned to them) has been even ever since the last cut-point ($O=\emptyset$) was reached;
it intuitively ensures that the respective vertices are finite.

\subsection{Tight Level Rankings}
Friedgut, Kupferman, and Vardi~\cite{FKV/06/tighter} improved this complementation technique by exploiting the observation that the true ranks of the run DAG $\mathcal G_\alpha$ of a rejected $\omega$-word~$\alpha$ are eventually always tight.
A level ranking $f:Q \rightarrow \omega$ is called \emph{tight}, if it has an odd rank~$r$, and is onto the odd numbers $\{1,3,\ldots,r\}$ up to its rank $r$, and \emph{$S$-tight}, if its restriction to $S$ is tight and if it maps all states not in $S$ to $1$ ($f(q)=1 \;\, \forall q \in Q\smallsetminus S$).

\begin{proposition}
\cite{FKV/06/tighter}
\label{prop:tight}
For every run DAG $\mathcal G_\alpha$ with finite rank $r$, it holds that
\begin{itemize}
\item $r$ is odd, and
\item there is a level $l\geq 0$ such that, for all levels $l' \geq l$ and all odd ranks $o \leq r$, there is a node $(q,l') \in \mathcal G_\alpha$ with rank $o$ in $\mathcal G_\alpha$.
\qed
\end{itemize}
\end{proposition}

This immediately follows from what was said in Subsection~\ref{ssec:dag} on reaching a fixed point after $n+1$ iterations:
If the rank $\mathcal G_\alpha$ is $r$ then, for every odd number $o\leq r$, almost all levels contain a vertex with rank $o$, and assuming that $r$ is even implies that all vertices of ${\mathcal G_\alpha}^r$ are finite, which in turn implies that only finitely many levels contain a vertex with rank $r$ and hence leads to a contradiction.

Using this observation, the construction from Proposition~\ref{prop:basic} can be improved by essentially replacing $\mathcal R$ by the set
$\mathcal T = \{f \in \mathcal R \mid f \mbox{ is tight}\}$
of tight level rankings.
While the size $|\mathcal R|=(2n+1)^n$ of $\mathcal R$ is in $\theta\big((2n)^n\big)$, the size of $\mathcal T$,
\[\tight(n)=|\mathcal T|,\]
is much smaller.
Building on an approximation of Stirling numbers of the second kind by Temme~\cite{Temme/93/Stirling}, Yan and Friedgut, Kupferman, and Vardi~\cite{FKV/06/tighter,Yan/08/lowerComplexity} showed that $\tight(n)$ can be approximated by $(\kappa n)^n$ for a constant $\kappa\approx 0.76$, that is, they showed
\[\kappa=\lim\limits_{n\rightarrow\infty}\frac{\sqrt[n]{\tight(n)}}{n}\approx 0.76.\]

Friedgut, Kupferman, and Vardi~\cite{FKV/06/tighter} use this observation---together with other improvements---for an improved complementation algorithm that produces a complement automaton with approximately $(0.96\,n)^n$ states~\cite{FKV/06/tighter}.

Yan~\cite{Yan/08/lowerComplexity} showed for full automata---a family of automata that has exactly one accepting state, and an alphabet that encodes the possible transitions between the states of the automaton---that every nondeterministic \buchi\ automaton that accepts the complement language of a full automaton with $n+1$ states must have $\Omega\big(\tight(n)\big)$ states.

\begin{proposition}
\cite{Yan/08/lowerComplexity}
\label{prop:lower}
A nondeterministic \buchi\ automaton that accepts the complement language of a full \buchi\ automaton with $n$ states has $\Omega\big(\tight(n-1)\big)$ states.
\qed
\end{proposition}

\section{Efficient \buchi\ Complementation}
\label{sec:efficient}

To optimize the construction from Proposition~\ref{prop:basic}, we turn not only to tight functions (c.f.~\cite{FKV/06/tighter}), but also refine the cut-point construction.
While the cut-point construction of Proposition~\ref{prop:basic} tests \emph{concurrently} for all even ranks if a path has finite even rank, we argue that it is much cheaper to test this property \emph{turn wise} for all even ranks individually.
As a result, the overall construction becomes more efficient and meets, modulo a small polynomial factor in $O(n^2)$, the lower bound recently established by Yan~\cite{Yan/08/lowerComplexity}.

\subsection{Construction}
\label{ssec:comp}
The obtained state space reduction of the proposed construction compared to~\cite{FKV/06/tighter} is due to an efficient cut-point construction in combination with the restriction to tight rankings.
The improved cut-point construction is inspired by the efficient translation from generalized to ordinary \buchi\ automata.
Indeed, the acceptance condition that no trace has an \emph{arbitrary} even rank, which is reflected by the straight-forward acceptance condition of previous algorithms~\cite{Kupferman+Vardi/01/Weak,FKV/06/tighter}, can be replaced by an acceptance condition, which~only rules out that some trace has a particular even rank, but does so for all potential even ranks.

Checking the condition for a particular even rank allows for focusing on exactly this rank in the cut-point construction, which led to a significant cut in the size of the resulting state space.
While this approach cannot be taken if we literally use a generalized \buchi\ condition, the idea of cyclically considering the relevant even ranks proves to be feasible.

\begin{construction}
For a given nondeterministic \buchi\ automaton $\mathcal A=(\Sigma,Q,I,\delta,F)$ with $n=|Q|$ states, let $\mathcal C=(\Sigma,Q',I',\delta',F')$ denote the nondeterministic \buchi\ automaton with
\begin{itemize}
\item $Q' = Q_1 \cup Q_2$ with
\begin{itemize}
\item $Q_1 = 2^Q$ and
\item $Q_2 = 
\{(S,O,f,i) \in 2^Q\times 2^Q \times \mathcal T \times \{0,2,\ldots,2n-2\} \mid
\newline \hspace*{1pt} \hfill f \mbox{ is $S$-tight, } O \subseteq S \mbox{ and }  \exists i \in \omega.\  O \subseteq f^{-1}(2i)\}$,
\end{itemize}

\item $I' = \{I\}$,

\item $\delta' = \delta_1 \cup \delta_2 \cup \delta_3$ for

\begin{itemize}
\item $\delta_1: Q_1 \times \Sigma \rightarrow 2^{Q_1}$ with $\delta_1(S,\sigma)= \{\delta(S,\sigma)\}$,
\item \mbox{$\delta_2: Q_1 \times \Sigma \rightarrow 2^{Q_2}$ with $(S',O,f,i)\in \delta_2(S,\sigma)$
$\Leftrightarrow S'{=} \delta(S,\sigma)$, $O{=}\emptyset$, and $i{=}0$,}
\item $\delta_3: Q_2 \times \Sigma \rightarrow 2^{Q_2}$ with $(S',O',f',i')\in \delta_3\big((S,O,f,i),\sigma\big)$
\newline $\Leftrightarrow \, S'= \delta(S,\sigma)$, $f'\leq_\sigma^S f$, $\rank(f)=\rank(f')$, and
\begin{itemize}
\item $i' = (i + 2) \mod (\rank(f') +1)$ and $O' = {f'}^{-1}(i')$ if $O = \emptyset$ or
\item $i'=i$ and $O' =\delta(O,\sigma) \cap {f'}^{-1}(i)$ if $O \neq \emptyset$, respectively, and
\end{itemize}
\end{itemize}

\item $F' = \{\emptyset\} \; \cup \;(2^Q\times \{\emptyset\} \times \mathcal T \times \omega) \cap Q_2 $.
\end{itemize}
\end{construction}

The complement automaton $\mathcal C$ operates in two phases.
In a first phase it only traces the states reachable in $\mathcal A$ upon a finite input sequence.
In this phase, only the states in $Q_1$ and the transition function $\delta_1$ are used.
In the special case that $\mathcal A$ rejects an $\omega$-word $\alpha$ because $\mathcal A$ has no run on $\alpha$, $\mathcal C$ accepts by staying forever in phase one, because $\{\emptyset\}$ is final.

$\mathcal C$ intuitively uses its nondeterministic power to guess a point in time where all successive levels are tight and have the same rank.
At such a point, $\mathcal C$ traverses from $Q_1$ to $Q_2$, using a transition from $\delta_2$.
Staying henceforth in $Q_2$ (using the transitions from $\delta_3$), $\mathcal C$ intuitively verifies turn wise for all potential even ranks $e$ that no path has this particular even rank~$e$.
For a particular rank~$e$, it suffices to trace the positions on traces with unchanged rank~$e$ (hence $O \subseteq f^{-1}(e)$), and to 
cyclically update the designated even rank after every cut-point.

\subsection{Correctness}
To show that the automaton $\mathcal C$ from the construction introduced in the previous subsection recognises the complement language of $\mathcal A$, we first show that the complement language of $\mathcal C$ contains the language of $\mathcal A$ ($\mathcal{L(A)} \subseteq \overline{\mathcal{L(C)}}$), and then that the complement language of $\mathcal A$ is contained in the language of $\mathcal C$ ($\overline{\mathcal{L(A)}} \subseteq \mathcal{L(C)}$).

\begin{lemma}
\label{lem:incl1}
If a given nondeterministic \buchi\ automaton $\mathcal A=(\Sigma,Q,I,\delta,F)$ accepts an $\omega$-word $\alpha: \omega \rightarrow \Sigma$, then $\alpha$ is rejected by the automaton $\mathcal C=(\Sigma,Q',I',\delta',F')$. ($\mathcal{L(A)} \subseteq \overline{\mathcal{L(C)}}$)
\end{lemma}

\begin{proof}
Let $\rho:\omega\rightarrow Q$ be an accepting run of $\mathcal A$ on $\alpha$.
First, $\rho'=S_0,S_1,S_2\ldots$ with $S_0=I$ and $S_{i+1}=\delta\big(S_i,\alpha(i)\big)$ for all $i\in \omega$ is no accepting run of $\mathcal C$, because $\rho(i)\in S_i$ and hence no state of $\rho'$ is accepting ($\emptyset\neq S_i\notin F'$ for all $i\in \omega$).
Let now
\[\rho'=S_0,S_1,S_2,\ldots,S_p,(S_{p+1},O_{p+1},f_{p+1},i_{p+1}),(S_{p+2},O_{p+2},f_{p+2},i_{p+2}),\ldots\]
be a run of $\mathcal C$ on $\alpha$.
Again, we have that $\rho(j)\in S_j$ for all $j\in \omega$. Furthermore, the construction guarantees that $f_{j+1} \leq_{\alpha(j)}^S f_j$ holds for all $j>p$. The sequence
\[f_{p+1}(\rho(p+1)\big) \geq f_{p+2}(\rho(p+2)\big) \geq f_{p+3}(\rho(p+3)\big) \geq \ldots\]
is therefore decreasing, and stabilizes eventually. That is, there is a $k > p$ and a $v \leq 2n$ such that $f_{l}\big(\rho(l)\big)=v$ for all $l\geq k$.
Since $\rho$ is accepting, there is a position $l\geq k$ with $\rho(l) \in F$. Taking into account that $f_l$ is a level ranking, this implies that $f_{l}\big(\rho(l)\big)$---and hence $v$---is even.
Assuming that $\rho'$ is accepting, we can infer that, for some position $l>k$ which follows one of the first $n$ accepting states of $\rho'$ after position $k$, $i_l= v$ and $O_l=f_l^{-1}(v) \ni \rho(l)$.
It is now easy to show by induction that, for all $m\geq l$, $i_m=v$ and (using $f_m\big(\rho(m)\big)=v)$) $\rho(m) \in O_m \neq \emptyset$ hold true, which contradicts the assumption that $\rho'$ is accepting.
\end{proof}

To proof the second lemma, $\overline{\mathcal{L(A)}} \subseteq \mathcal{L(C)}$, we use Propositions~\ref{prop:dag} and~\ref{prop:tight} to infer that the run DAG $\mathcal G_\alpha$ of an $\omega$-word rejected by $\mathcal A$ is either finite or has odd bounded rank and only finitely many non-tight level rankings.
We use this to build an accepting run of $\mathcal C$ on~$\alpha$.

\begin{lemma}
\label{lem:incl2}
For a nondeterministic \buchi\ automaton $\mathcal A=(\Sigma,Q,I,\delta,F)$, the automaton $\mathcal C=(\Sigma,Q',I',\delta',F')$ accepts an $\omega$-word $\alpha: \omega \rightarrow \Sigma$ if $\alpha$ is rejected by $\mathcal A$. ($\overline{\mathcal{L(A)}} \subseteq \mathcal{L(C)}$)
\end{lemma}

\begin{proof}
If $\alpha: \omega \rightarrow \Sigma$ is rejected by $\mathcal A$, then the run DAG $\mathcal G_\alpha$ has bounded rank by Proposition~\ref{prop:dag}, and by Proposition~\ref{prop:tight} almost all levels of $\mathcal G_\alpha$ have tight level rankings with the same rank $r$.
For the special case that the rank of all vertices of $\mathcal G_\alpha$ is $0$, that is, if all vertices of $\mathcal G_\alpha$ are finite,
$\rho'=S_0,S_1,S_2\ldots$
with $S_0=I$ and $S_{i+1}=\delta\big(S_i,\alpha(i)\big)$ for all $i\in \omega$ is an accepting run of $\mathcal C$ on $\alpha$.

If $\mathcal G_\alpha$ contains an infinite vertex, then we fix a position $p \in \omega$ such that the rank of all levels $p' \geq p$ of $\mathcal G_\alpha$ is $r$ and tight for some (odd) $r \geq 1$.
We now consider a run
\[\rho'=S_0,S_1,S_2,\ldots,S_p,(S_{p+1},O_{p+1},f_{p+1},i_{p+1}),(S_{p+2},O_{p+2},f_{p+2},i_{p+2}),\ldots
\mbox{ of $\mathcal C$ on $\alpha$ with}\]
 \begin{itemize}
\item $S_0=I$, $O_{p+1}=\emptyset$, and $i_{p+1}=0$,
\item $S_{j+1}=\delta\big(S_j,\alpha(j)\big)$ for all $j\in \omega$, and
\item $O_{j+1} = f_{j+1}^{-1}(i_{j+1})$ if $O_j=\emptyset$ or
\newline $O_{j+1} = \delta\big(O_j,\alpha(j)\big) \cap f_{j+1}^{-1}(i_{j+1})$ if $O_j\neq\emptyset$, respectively, for all $j>p$,
\item $f_j$ is the $S_j$-tight level ranking that maps each state $q{\in} S_j$ to the rank of $(q,j)$~$\forall j{>}p$, 
\item $i_{j+1}=i_j$ if $O_j\neq\emptyset$ or
\newline $i_{j+1}=(i_j + 2) \mod (\rank(f)+1)$ if $O_j=\emptyset$, respectively, for all $j>p$.
\end{itemize}

$\rho'$ is obviously a run of $\mathcal C$ on $\alpha$.
To show that $\rho'$ is accepting, we have to show that $O_j$ is empty infinitely many times.
Let us assume that this is not the case; that is, let us assume that there is a last element $\rho'(j)$ with $O_j = \emptyset$ and $O_k \neq \emptyset$ for all $k>j$.
(Note that $O_{p+1}=\emptyset$ is empty.)
Then we have that $i_k = i_{j+1}$ for all $k>j$, and it is easy to show by induction for all $k > j$ that $O_k\times \{k\}$ is the set of states reachable in ${\mathcal G_\alpha}^{i_{j+1}}$ from some state in $O_{j+1} \times \{j+1\}$.
But since the rank of all states in $O_{j+1} \times \{j+1\}$ is $i_{j+1}$ (and thus even), all of these states are finite in ${\mathcal G_\alpha}^{i_{j+1}}$, which implies that there are only finitely many states reachable in ${\mathcal G_\alpha}^{i_{j+1}}$ from $O_{j+1} \times \{j+1\}$, and thus contradicts the assumption that $O_k\neq\emptyset$ is non-empty for all $k>j$.
\end{proof}

The two lemmata of this subsection immediately imply the claimed language complementation:

\begin{corollary}
For a given nondeterministic \buchi\ automaton $\mathcal A=(\Sigma,Q,I,\delta,F)$, the automaton $\mathcal C$ resulting from the construction introduced in Subsection~\ref{ssec:comp} recognises the complement language of $\mathcal A$.
($\mathcal{L(C)} = \overline{\mathcal{L(A)}}$) \qed
\end{corollary}

\subsection{Complexity}

The costly part in previous approaches~\cite{Kupferman+Vardi/01/Weak,FKV/06/tighter} that the proposed method avoids (de facto, although not technically), is a subset construction in addition to the level rankings.
Avoiding the subset construction results in a state space reduction by a factor exponential in the size of the automaton $\mathcal A$, and even to upper bounds comparable to the established lower bounds~\cite{Yan/08/lowerComplexity}.

The subset construction is de facto avoided, because it can be encoded into the ranking function once we allow for a slightly enlarged set of output values.
For example, we could map all states not in $S$ to $-1$, and all states in $O$ to $-2$. Following this convention, the first two elements of every tuple in $Q_2$ could be pruned.
(They remain explicit in the construction because this representation is more comprehensible, and outlines the connection to the older constructions of Friedgut, Kupferman, and Vardi~\cite{Kupferman+Vardi/01/Weak,FKV/06/tighter}.

\begin{theorem}
\label{theo:upper}
For a given nondeterministic \buchi\ automaton $\mathcal A$ with $n$ states, the automaton $\mathcal C$ has $O\big(\tight(n+1)\big)$ states.
\end{theorem}

\begin{proof}
$Q_1$ is obtained by a simple subset construction, and hence $Q_1 \in O(2^n)$, which is a small subset of $O\big(\tight(n+1)\big)$.
For a fixed $i$, a state $(S,O,f,i)$ with an $S$-tight level ranking $f$ of rank $r$ can be represented by a function $g:Q\rightarrow \{-2,-1,\ldots,r\}$ that maps every state $q \in O$ to $g(q)=-1$, every state $q \in Q \smallsetminus S$ not in $S$ to $g(q)=-2$, and every other state $q\in S\smallsetminus O$ to $g(q)=f(q)$.
Every such function $g$ either has a domain $g(Q) \subseteq \{0,1,\ldots,r\}$ and is onto $\{1,3,\ldots,r\}$, or is a function to $\{-2,-1,\ldots,r\}$ and onto $\{-1\}\cup\{1,3,\ldots,r\}$ or $\{-2\}\cup\{1,3,\ldots,r\}$.
The size of all three groups of functions is hence in $O\big(\tight(n)\big)$ (for a fixed $i$), which results in an overall size in $O\big(\tight(n+1)\big)=O\big(n\cdot\tight(n)\big)$.
\end{proof}

Together with Proposition~\ref{prop:lower}, this establishes tight bounds for \buchi\ complementation:

\begin{corollary}
The minimal size of the state space of a nondeterministic \buchi\ automaton that accept the complement language of a nondeterministic \buchi\ automaton with $n$ states is in $\Omega\big(\tight(n-1)\big)$ and $O\big(\tight(n+1)\big)$.
\qed
\end{corollary}

\section{Reduced Average Outdegree}

A flaw in the construction presented in Section~\ref{sec:efficient} is that it is optimal only with respect to the state space of the automata.
In~\cite{Kupferman+Vardi/01/Weak}, Kupferman and Vardi discuss how to reduce the number of edges such that the bound on the number of edges becomes trilinear in the alphabet size, the bound on the number of states, and the rank of the resulting automaton (see also~\cite{GKSV/03/Complementing}).
In this section we improve the construction from the previous section such that the bound on the number of edges is merely bilinear in the bound on the number of states and the size of the input alphabet.
The technique can be adapted to generally restrict the outdegree to $|\delta(q,\sigma)| \leq 2$ when level rankings are not required to be tight.

\subsection{Construction}
\label{subs:ConD}
The automaton $\mathcal C$ obtained from the construction described in Section~\ref{sec:efficient} operates in two phases.
In a first phase, it stays in $Q_1$ and only tracks the reachable states of the \buchi\ automata $\mathcal A$ it complements.
It then guesses a point $p \in \omega$ such that all levels $j>p$ of $\mathcal G_\alpha$ have a tight level ranking to transfer to $Q_2$.

We improve over the construction of Section~\ref{sec:efficient} by restricting the number of entry points to $Q_2$ from $O\big(\tight(n)\big)$ to $O(n!)$, and by restricting the number of outgoing transitions $|\delta(q,\sigma)|\leq 2$ for all states $q \in Q_2$ and input letters $\sigma \in \Sigma$ to two.
The latter is achieved by allowing only the successor $(S,O,f,i)\in \delta_3(q,\sigma)$ with a point wise maximal function $f$ (the $\gamma_3$-transitions) or with a function $f$ that is maximal among the final states $(S,O,f,i)\in \delta_3(q,\sigma) \cap F$ among them (the $\gamma_4$-transitions).
If such elements exist, then they are unique.

The first restriction is achieved by restricting $\delta_2$ to states $(S,O,f,i)$ for which $f$ is maximal with respect to $S$.
We call an $S$-tight level ranking $f$ with rank $r$ \emph{maximal with respect to $S$} if it maps all final states $q \in F \cap S$ in $S$ to $r-1$, exactly one state to every odd number $o<r$ smaller than $r$ ($|f^{-1}(o)|=1$) and all remaining states of $S$ to $r$, and denote the set of tight rankings that are maximal with respect to $S$ by $\mathcal M_S= \{f \in \mathcal T \mid f \mbox{ is maximal with respect to } S\}$.

As there are only $|Q_1|\leq 2^n$ states in $Q_1$, the impact of their high outdegree is outweighed by the small outdegree ($|\delta(q,\sigma)|\leq 2$) of the remaining $|Q_2|\in O\big(\tight(n+1)\big)$ states.

\begin{construction}
For a given nondeterministic \buchi\ automaton $\mathcal A=(\Sigma,Q,I,\delta,F)$ with $n=|Q|$ states, let $\mathcal D=(\Sigma,Q',I',\gamma,F')$ denote the nondeterministic \buchi\ automaton with~$Q'$, $I'$, $F'$, $Q_1$, $Q_2$, and $\delta_1$ as in the construction from Section~\ref{sec:efficient}, and with $\gamma = \delta_1 \cup \gamma_2 \cup \gamma_3 \cup \gamma_4$~for
\begin{itemize}
\item $\gamma_2: Q_1 \times \Sigma \rightarrow 2^{Q_2}$ with $(S',O,g,i)\in \gamma_2(S,\sigma) \Leftrightarrow (S',O,g,i)\in \delta_2(S,\sigma)$ and $g \in \mathcal M_{S'}$,

\item $\gamma_3: Q_2 \times \Sigma \rightarrow 2^{Q_2}$ with $\gamma_3\big((S,O,f,i),\sigma\big)= \{\maxi_g\big\{(S',O',g,i'){\in} \delta_3\big((S,O,f,i),\sigma\big)\big\}$,

\item $\gamma_4: Q_2 \times \Sigma \rightarrow 2^{Q_2}$ with $(S',O'',g',i')\in \gamma_4\big((S,O,f,i),\sigma\big)$ if
\newline \hspace*{1pt} \hfill $(S',O',g,i')\in \gamma_3\big((S,O,f,i),\sigma\big)$, $O''=\emptyset$, $i'\neq 0 \vee O' = \emptyset$, and
\newline \hspace*{1pt} \hfill $g'(q)=g(q)-1$ for all $q \in O'$ and $g'(q)=g(q)$ otherwise, \hspace*{4.4mm}
\end{itemize}
where $\maxi_g\big\{(S',O',g,i')\in \delta_3\big((S,O,f,i),\sigma\big)\big\}$ selects the unique element with a (point wise) maximal function $g$.
The supremum over all function obviously exists, but it is not necessarily tight. (In this case, $\maxi_g$ returns the empty set.)
Since $\delta_3$ is strict with respect to the selection of the other three elements $S'$, $O'$ and $i'$ for a fixed ranking function, the mapping of $\gamma_3$ consists only of singletons and the empty set.
\end{construction}

While $\gamma_3$ selects a maximal successor, $\gamma_4$ selects a maximal final successor, which only requires to decrease the value assigned to the states in $O'$ by the ranking function $g$ by one.
(Which cannot be done if their value is already $0$, hence the restriction $i\neq 0$ or $O'=\emptyset$.)
The tightness of $g'$ is then inherited from the tightness of $g$.

\subsection{Correctness}

While it is clear that the language of $\mathcal D$ is contained in the language of $\mathcal C$, the converse is less obvious. To prove $\mathcal{L(C)}\subseteq \mathcal{L(D)}$, we show that any $\omega$-word $\alpha$ rejected by $\mathcal A$ will be accepted by $\mathcal D$ by exploiting the ``standard'' run $\rho'$ of $\mathcal C$ on $\alpha$ from the proof of Lemma \ref{lem:incl2} to build an accepting run $\rho''$ of $\mathcal D$ on $\alpha$.

\begin{proposition}
\label{prop:trans}
For a given nondeterministic \buchi\ automaton $\mathcal A=(\Sigma,Q,I,\delta,F)$, the automaton $\mathcal D$ resulting from the construction introduced in Subsection~\ref{subs:ConD} accepts an $\omega$-word $\alpha: \omega \rightarrow \Sigma$ if and only if $\alpha$ is rejected by $\mathcal A$. ($\mathcal{L(D)} = \overline{\mathcal{L(A)}}$)
\end{proposition}

\begin{proof}
We show $\mathcal{L(D)} = \overline{\mathcal{L(A)}}$ by demonstrating $\overline{\mathcal{L(A)}} \subseteq \mathcal{L(D)} \subseteq \mathcal{L(C)} \subseteq \overline{\mathcal{L(A)}}$, where the second inclusion is implied by the fact that every (accepting) run of $\mathcal D$ is also an (accepting) run of $\mathcal C$, and the third inclusion is shown in Lemma~\ref{lem:incl1}.

To demonstrate $\overline{\mathcal{L(A)}} \subseteq \mathcal{L(D)}$, we reuse the proof of Lemma~\ref{lem:incl2} to obtain an accepting run $\rho'=S_0,S_1,S_2,\ldots,S_p,(S_{p+1},O_{p+1},f_{p+1},i_{p+1}),(S_{p+2},O_{p+2},f_{p+2},i_{p+2}),\ldots$ for $\mathcal C$, where $f_j(q)$ is the rank of $(q,j)$ in $\mathcal G_\alpha$ for all $j>p$ and $q \in S_j$.
(If $\mathcal A$ has no run on $\alpha$, then $\mathcal D$ has the same accepting standard run on $\alpha$ that stays in $Q_1$ as $\mathcal C$.)

Let us pick an $S_{p+1}$-tight ranking function $g_{p+1}$ with the same rank as $f_{p+1}$ that is maximal with respect to $S_{p+1}$.
We show that we then can construct the run
\[\rho''=S_0,S_1,S_2,\ldots,S_p,(S_{p+1},O_{p+1}',g_{p+1},i_{p+1}'),(S_{p+2},O_{p+2}',g_{p+2},i_{p+2}'),\ldots\]
of $\mathcal D$ on $\alpha$ that satisfies $O_{p+1}=\emptyset$, and $i_{p+1}=0$ and, for all $j>p$,

\begin{itemize}
\item $(S_{j+1},O_{j+1}',g_{j+1},i_{j+1}') \in \gamma_4\big((S_j,O_j',g_j,i_j'),\alpha(j)\big)$ if $O_{j+1} = \emptyset$ and $i_{j+1}=i_{j+1}'$
\newline (note that $i_{j+1}'$ does not depend on taking the transition from $\gamma_3$ or $\gamma_4$), and 
\item $(S_{j+1},O_{j+1}',g_{j+1},i_{j+1}') \in \gamma_3\big((S_j,O_j',g_j,i_j'),\alpha(j)\big)$ otherwise.
\end{itemize}

To show by induction that $g_j \geq f_j$ holds true for all $j>p$ (where $\geq$ is the point wise comparison), we strengthen the claim by claiming additionally that if, for some position $k>p$, $\rho'(k)$ is a final state, $i_{k+1}'=i_{k+1}$ holds true, and $k'>k$ is the next position for which $\rho''(k')$ is final, then $q \in O_j'$ implies $q \in O_j \vee g_j(q) > f_j(q)$ for all $k < j \leq k'$.

For $j=p+1$ this holds trivially (basis).
For the induction step, let us first consider the case of $\gamma_3$-transitions.
Then $g_{j+1}\geq f_{j+1}$ is implied, because $g_{j+1}$ is maximal among the $S_{j+1}$-tight level rankings $\leq_{\alpha(j)}^{S_j} g_j$, and $g_j\geq f_j$ holds by induction hypothesis.
If $\rho'(j)$ is a final state and $i_{j+1}'=i_{j+1}$ holds true, then $O_{j+1} = {f_j}^{-1}(i_{j+1})$, and hence $q \in O_j'$ implies $g_{j+1}(q)=i_{j+1}'=i_{j+1}$, which implies $q \in O_{j+1} \vee g_{j+1}(q) > f_{j+1}(q)$ (using $g_{j+1}\geq f_{j+1}$).

If a $\gamma_4$-transition is taken, then taking a $\gamma_3$-transition implied $g_{j+1}\geq f_{j+1}$ and $g_{j+1}(q)=i_{j+1}'=i_{j+1} \Rightarrow g_{j+1}(q) > f_{j+1}(q)$ (note that $O_{j+1}$ is empty) by the previous argument.
This immediately implies $g_{j+1}\geq f_{j+1}$ for the $\gamma_4$-transition.
Consequently, $O_{j+1} = {f_j}^{-1}(i_{j+1})$ (which holds as $\rho'(j)$ is final)
entails that $q \in O_j'$ implies $q \in O_{j+1} \vee g_{j+1}(q) > f_{j+1}(q)$.

It remains to show that all functions $g_j$ are $S_j$-tight level rankings.
To demonstrate this, let $q\in S_{p+1}$ be a state of the automaton $\mathcal A$ such that $g_{p+1}(q)=o$ is the rank of $(q,p+1)$ in $\mathcal G_\alpha$ for an odd number $o \leq r$.
(Such a state exists for every odd number $o \leq r$ by construction.)
Since $(q,p+1)$ is endangered but not finite in ${\mathcal G_\alpha}^o$, all nodes $(q',j)$ with $j> p$ reachable from $(q,p+1)$ in ${\mathcal G_\alpha}^o$ form an infinite connected sub-DAG of ${\mathcal G_\alpha}^o$, all of whose nodes have rank $o$. (Which, by the proof in Lemma~\ref{lem:incl2}, entails that $f_j(q')=o$ holds for every vertex $(q',j)$ of this sub-DAG.)
By definition of $\rho''$, it is easy to show by induction that $g_{j}(q')\leq o$ holds for all of these nodes $(q',j)$.
As we have just demonstrated $g_j(q')\geq f_j(q')$, this entails $g_j(q') = o$.
Since $o$ can be any odd number less or equal to the rank $r$ of $\mathcal G_\alpha$, and since there is, for every $j>p$, some vertex $(q',j)$ reachable from $(q,p+1)$ in ${\mathcal G_\alpha}^o$, $g_j$ is an $S_j$-tight level ranking for all $j>p$.

Finally, the assumption that $O_j'$ is empty only finitely many times implied that there was a last position $k$ such that $O_k'$ is empty. But this implies that $i_j$ is stable for all $j>k$, and within the next $n$ visited fixed points in $\rho'$ there is one that refers to this $i_{k+1}$.
By construction of $\rho''$, this position is a final state in $\rho''$, too. $\lightning$
\end{proof}

\subsection{Complexity}
Extending the tight bound of Section~\ref{sec:efficient} for the state space to a tight bound on the size of the complement automaton is simple:
The mappings of $\delta_1$, $\gamma_3$, and $\gamma_4$ consist of singletons or the empty set, such that only the size of $\gamma_2$ needs to be considered more closely.

\begin{theorem}
\label{theo:upperSize}
For a given nondeterministic \buchi\ automaton $\mathcal A$ with $n$ states and an alphabet of size $s$, the automaton $\mathcal D$ has size $O\big(s\,\tight(n+1)\big)$.
\end{theorem}

\begin{proof}
For all $S \in Q_1$ and $\sigma \in \Sigma$, we have that $\gamma_2(S,\sigma) = \{S'\} \times \{\emptyset\} \times \mathcal M_{S'} \times \{0\}$ for $S' = \delta(S)$.
Thus,  $|\gamma_2(S,\sigma)| = |\mathcal M_{S'}|$, which can be estimated by $\sum_{i=1}^{m} \frac{m!}{i!}$ for $m=|S' \smallsetminus F|$, which is in $O(n!)$.
Thus
$\sum\limits_{S \in Q_1,\ \sigma \in \Sigma}|\gamma_2(S,\sigma)| \in O(s\, 2^n \, n!) \subsetneq o\big(s\,\tight(n)\big)$
holds true.

($2^n \, n! \approx \big(\frac{2n}{e}\big)^n \approx (0.74\,n)^n$, whereas $\tight(n) \approx (0.76\, n)^n$.)
The claim thus follows with Theorem~\ref{theo:upper}, and $|\delta_1(q_1,\sigma)|=1$ and $|\gamma_3(q_2,\sigma)|,|\gamma_4(q_2,\sigma)|\leq 1$ for all $q_1 \in Q_1$, $q_2 \in Q_2$ and $\sigma \in \Sigma$.
\end{proof}

Together with Proposition~\ref{prop:lower}, this establishes tight complexity bounds for \buchi\ complementation:

\begin{corollary}
The complexity of complementing nondeterministic \buchi\ automata with $n$~states is in $\Omega\big(\tight(n-1)\big)$ and $O\big(\tight(n+1)\big)$.
The discussed complementation technique is therefore optimal modulo a small polynomial factor in $O(n^2)$.
\qed
\end{corollary}

\section{Discussion}

This paper marks the end of the long quest for the precise complexity of the \buchi\ complementation problem.
It shows that the previously known lower bound is sharp, which is on one hand surprising, because finding tight lower bounds is generally considered the harder problem, and seems on the other hand natural, because Yan's lower bound builds on the concept of tight level rankings alone~\cite{Yan/08/lowerComplexity}, while the previously known upper bound~\cite{FKV/06/tighter} incorporates an additional subset construction and builds on estimations on top of this, leaving the estimations of the lower bound the simpler concept of the two.

Similar to the complexity gap in \buchi\ complementation twenty years ago, the complexity of \buchi\ determinization is known to be in $n^{\theta(n)}$, but there is also an $n^{\theta(n)}$ gap between the upper~\cite{Schewe/09/determinise} and lower~\cite{Yan/08/lowerComplexity} bound.
Tightening the bounds for \buchi\ determinization appears to be the natural next step after the introduction of an optimal \buchi\ complementation algorithm.

\newcommand{\etalchar}[1]{$^{#1}$}

\end{document}